\documentclass{article}

\usepackage[english]{babel}

\usepackage[letterpaper,top=2cm,bottom=2cm,left=3cm,right=3cm,marginparwidth=1.75cm]{geometry}

\usepackage{amsmath}
\usepackage{amsthm}
\theoremstyle{definition}
\newtheorem{definition}{Definition}
\newtheorem{lemma}[definition]{Lemma}
\newtheorem{observation}[definition]{Observation}
\newtheorem{remark}[definition]{Remark}
\newtheorem{theorem}[definition]{Theorem}
\usepackage{amssymb}
\usepackage{csquotes}
\usepackage{graphicx}
\usepackage{xcolor}
\usepackage[colorlinks=true, allcolors=blue]{hyperref}
\usepackage{authblk}

\newcommand{\height}{\mathrm{ht}}
\newcommand{\A}{\mathcal{A}}

\title{On smallest synchronizing terms over constant alphabets}
\author[1,2]{Luisa Herrmann\thanks{\protect\url{luisa.herrmann@tu-dresden.de}}}
\author[1]{Richard Mörbitz\thanks{\protect\url{richard.moerbitz@tu-dresden.de}}}

\affil[1]{TU Dresden}
\affil[2]{ScaDS.AI Dresden/Leipzig}

\begin{document}
\maketitle

\begin{abstract}
We show a subexponential lower bound on the reset threshold of synchronizing deterministic finite tree automata (DTA) over alphabets of just two symbols.
This significantly improves the previous one, which was quadratic in the number of states.
Our result also narrows the gap towards the lower bound for DTA over alphabets that grow linearly with the number of states, as well as the best known upper bound, both of which are currently exponential.
\end{abstract}

\section{Introduction}

In the field of deterministic finite (string) automata (DFA), the notion of synchronization is a prominent and actively investigated topic.
A string $w$ is synchronizing for a given DFA $\A$ if there exists a fixed state $q$ such that for every state $q'$ of $\A$, when $\A$ starts reading $w$ in $q'$, it ends in $q$~\cite{Liu1963,Cerny1964}.
Applications include classical problems of automation and robotics, such as orienting parts on a conveyor belt~\cite{Natarajan1986} and navigating a robot from an unknown position to a known position~\cite{BealPerrin2014} (also cf.~\cite{AdlerWeiss1970,AdlerGoodwynWeiss1977}), but also communications in the form of self-synchronizing codes~\cite{Liu1963,PerrinSchutzenberger1992}.

One can easily see that not every DFA has a synchronizing word.
Whether a synchronizing word exists for a given DFA can be determined in cubic time wrt.\@ its number of states~\cite{Eppstein1990}.
A DFA that has a synchronizing word is called synchronizing; thus, the class of synchronizing DFA is a proper subclass of the class of all DFA.
For synchronizing DFA, their shortest synchronizing word is of particular interest,
and its length is referred to as the reset threshold of the DFA.
While the problem of finding the shortest synchronizing word of a DFA is known to be NP-hard~\cite{Rystsov1983,Eppstein1990}, the achievable reset threshold (wrt.\@ the number of states of the automaton) is still ongoing research.
Černý~\cite{Cerny1964} constructed a family of DFA with $n$ states whose shortest synchronizing word has length $(n-1)^2$,
and it was conjectured that this lower bound is tight~\cite{Starke1966,CernyPirickaRosenauerova1971}.
This conjecture, widely known as the \emph{Černý conjecture}, turned out to be one of the longest-standing open problems in automata theory:
to this day, only a cubic upper bound is known~\cite{Starke1966,Pin1983,Frankl1982}, and the only advances were able to improve the bound on the level of constants~\cite{Szykula2018,Shitov2019}.
We note that these bounds on the reset threshold only hold if the DFA is total, i.e., for every state and symbol, there is exactly one successor state.
For partial automata, exponential lower bounds are known~\cite{deBondtDonZantema2019}.

Finite-state tree automata (DTA) are a natural generalization of DFA that introduce branching~\cite{GecsegSteinby1984}.
Rather than linear strings, where each symbol (except the last one) has exactly one successor, they recognize trees, where each symbol may have an arbitrary (but fixed for that symbol) number of successors.
Recently~\cite{BlazejJanousekPlachy2023,BlaJanPla26}, the synchronization problem was also investigated for DTA -- with the natural adaptation, i.e., instead of a synchronizing word, a tree (called \emph{synchronizing term}) is desired.
As counterpart of the word's length, the height of the tree is considered, i.e., the reset threshold of a DTA is the height of its smallest synchronizing term.
For trees, the cubic upper bound does not carry over from the string case.
Instead, the authors were only able to show a trivial upper bound of $2^n - n - 1$.
For lower bounds, it is clear that Černý's quadratic bound carries over (since every DFA can be considered as a DTA).
The authors were also able to show exponential lower bounds, but only if the alphabet grows with the number of states.
This is not in the spirit of Černý's original construction, which featured a constant alphabet (of size 2) independent of the number of states.
Thus, for families of DTAs over constant alphabets, a significant gap remains between the quadratic lower bound and the exponential upper bound on the height of the smallest synchronizing term.

Our contribution is to provide a new flexible construction for families of DTA over alphabets of just two symbols, which allows to improve the existing lower bounds on the reset threshold of DTA in multiple steps.
First, for every degree $d \ge 2$, we obtain a lower bound that grows in $\Omega(n^d)$;
already for $d=4$ this separates the reset threshold of DTA from that of DFA (which is upper bounded by $O(n^3)$).
Second, we obtain a subexponential lower bound whose growth is equivalent to Landau's function ($e^{(1+o(1))\sqrt{n \cdot \ln n}}$).

This paper is structured as follows.
We recall some basic notions from formal language theory and number theory in Section~\ref{sec:preliminaries}.
Section~\ref{sec:sync-fta} contains all definitions related to the synchronization of DTA.
In Section~\ref{sec:construction}, we present the construction used to show our improved bounds, which we will do in Section~\ref{sec:bounds}.

\section{Basic Notions}
\label{sec:preliminaries}

The set of natural numbers (including $0$) is denoted by $\mathbb N$ and we let $\mathbb N_+ = \mathbb N \setminus \{ 0 \}$.
For every $n \in \mathbb N$, we let $[n] = \{ 1, \dots, n \}$.
Thus, in particular, $[0] = \emptyset$.

The \emph{composition} of two functions $f_1\colon A \to B$ and $f_2 \colon B \to C$ is the function $(f_2 \circ f_1) \colon A \to C$ defined by $(f_2 \circ f_1)(a) = f_2(f_1(a))$ for every $a \in A$.
For every unary function $f\colon A \to A$ and $i \in \mathbb N$, we let $f^i\colon A \to A$ denote the $i$-fold composition of $f$.
Formally, $f^0$ is the identity and $f^{i+1}(a) = f(f^i(a))$ for every $i \in \mathbb N$ and $a \in A$.
Moreover, we let $f^{\le n}(a)=\bigcup_{i\in\{0,\dots,n\}} \{ f^i(a) \}$ for every $n \in \mathbb N$ and $f^*(a)=\bigcup_{i\in\mathbb N} \{ f^i(a) \}$.

In this work, we are a bit lenient with the notation of composing higher-arity functions, e.g., given a function $f_3\colon B^2 \to C$, we let $f_3(f_1, f_1) \colon A^2 \to C$ be the function defined for every $(a, a') \in A^2$ by $f_3(f_1, f_1)((a, a')) = f_3(f_1(a), f_1(a'))$.
Whenever we speak of applying such a composed function to a single argument, we mean that it is applied to a tuple consisting of $k$ copies of this argument where $k$ is the arity of the composed function.
Often we are not only interested in the functions themselves, but also in their (tree-like) syntactic shape of composition (commonly denoted as \emph{terms}).
We refer to this shape as a \emph{composite expression} in order to avoid confusion with the notions \enquote{tree} and \enquote{term} that occur in the context of synchronizing tree automata. When we speak about the height of such an expression, we mean the same concept as in the case of trees (see below). Intuitively, it describes how deeply nested the expression is.

\paragraph{Ranked alphabets} By an \emph{alphabet}, in this work often denoted by $\Sigma$, we mean a non-empty and finite set whose elements are called \emph{symbols}. Here, we assign to each symbol a \emph{rank} which will stand for the number of its successor nodes in a tree. Formally, we say that a \emph{ranked alphabet} is a tuple $(\Sigma,\mathrm{rk})$ consisting of an alphabet $\Sigma$ together with a rank function $\mathrm{rk}\colon\Sigma\to\mathbb{N}$. In the usual way, we abbreviate the ranked alphabet $(\Sigma,\mathrm{rk})$ by $\Sigma$ if $\mathrm{rk}$ is clear from the context and we let $\Sigma^{(n)}$ stand for $\mathrm{rk}^{-1}(n)$. Moreover, we often define $\mathrm{rk}$ implicitly by writing $\sigma^{(n)}$ instead of $\mathrm{rk}(\sigma)=n$ for a symbol $\sigma\in\Sigma$.

\paragraph{Strings, trees, and tree contexts} Given an alphabet $\Sigma$, each finite sequence $w=a_1\ldots a_n$ consisting of symbols $a_1,\ldots,a_n\in \Sigma$ is called a \emph{string} over $\Sigma$ (of length $n\in\mathbb{N}$). If we now let $\Sigma$ be a ranked alphabet, we can easily generalize from strings to trees by letting each involved symbol $a_i$ have not just one but $\mathrm{rk}(a_i)$ many successors.

If a symbol has rank $0$ and, thus, no successors in a tree, we call it a \emph{leaf} ($b_0$ in the picture below). We will often consider trees which do not terminate with leaves but have open ends (in the sense that one could plug in any other tree). In order to address these points, we will use a set $X=\{x_1,x_2,\ldots\}$ of \emph{variables}, which is disjoint from any other set in this work. Formally, the set $\mathrm{T}_\Sigma(X)$ of \emph{trees} over $\Sigma$ and $X$ is the smallest set $T$ such that $(i)$ $X\subseteq T$ and $(ii)$ for each $n\in\mathbb{N}$, $\sigma\in\Sigma^{(n)}$, and $\xi_1,\ldots,\xi_n\in T$ we have $\sigma(\xi_1,\ldots,\xi_n)\in T$. We call a tree $\xi\in \mathrm{T}_\Sigma(X)$ a \emph{context} if (i) no variable occurs twice and (ii) the left-to-right sequence of variables occurring in the tree is $x_1 x_2\ldots x_l$. A \emph{full context} is a context without leaves.

\begin{figure}[t]
    \centering
    \includegraphics[width=1\linewidth]{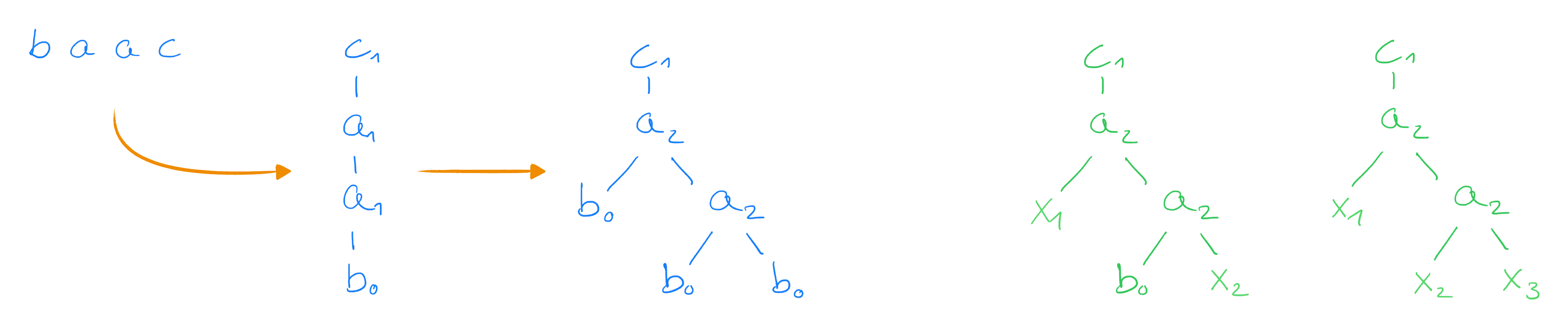}
    \caption{In the blue part, a string $w=baac$ over the alphabet $\Sigma_s=\{a,b,c\}$ is turned into a linear tree over the ranked alphabet $\Sigma_l=\{a_1^{(1)},b_0^{(0)},c_1^{(1)}\}$ where all symbols except the nullary leaf symbol are of rank 1. We can generalize this view to non-linear trees by allowing symbols of higher rank as e.g. with the ranked alphabet $\Sigma=\{a_2^{(2)},b_0^{(0)},c_1^{(1)}\}$. The left tree of the green part is a context, while the right tree is even a full context not having any leaf symbol.}
    \label{fig:trees}
\end{figure}

Given a tree $\xi\in\mathrm{T}_\Sigma(X)$ we define its \emph{height} as the length of the longest path (thinking of $\xi$ as a graph) from the root to a leaf symbol or variable, i.e., $\height(\xi)=0$ if $\xi\in\Sigma^{(0)}\cup X$ and $\height(\xi)=1+\mathrm{max}\{\height(\xi_1),\ldots,\height(\xi_n)\}$ if $\xi=\sigma(\xi_1,\ldots,\xi_n)$. The height of all trees in Figure \ref{fig:trees} is 3.

\paragraph{Number theory}
For two natural numbers $a$ and $b$ we say that $a$ \emph{divides} $b$, abbreviated by $a | b$, if there exists a natural number $c$ such that $a \cdot c = b$.
Let $k, \ell_1, \dots, \ell_k \in \mathbb N_+$.
The \emph{least common multiple of $\{ \ell_1, \dots, \ell_k \}$}, denoted by $\mathrm{lcm} \{ \ell_1, \dots, \ell_k \}$, is the smallest positive number $a$ such that $\ell_1 | a, \dots, \ell_k | a$.
Given $n \in \mathbb N$, Landau's function $g(n)$ is the largest least common multiple of all partitions of $n$, i.e.,
\[
    g(n) = \max \{ \mathrm{lcm} \{\ell_1, \dots, \ell_k\} \mid k \in \mathbb N_+, (\ell_1, \dots, \ell_k) \in (\mathbb N_+)^k, \ell_1 + \cdots + \ell_k \le n \} \text.
\]
While no precise closed form of $g(n)$ is known, it can be approximated as $g(n) = e^{(1+o(1)) \cdot \sqrt{n \cdot \ln n}}$~\cite{Lan1903}.
We note that the tuple $(\ell_1, \dots, \ell_k)$ that witnesses $g(n)$ is not necessarily unique.
However, for every $n \ge 5$ there exists one that has the properties required by our construction (cf.\@ Theorem~\ref{thm:landau}).

\begin{lemma}\label{lem:lower-bound-landau}
    For every $n \ge 5$ there exist $k \ge 2$ and $\ell_1, \dots, \ell_k \in \mathbb N_+$ such that $\ell_1 + \dots + \ell_k \le n$, $\mathrm{lcm} \{ \ell_1, \dots, \ell_k \} = g(n)$, and there is some $i \in [k]$ with $\ell_i \ge 2$.
\end{lemma}

\begin{proof}
    Let $n \ge 5$ and $m = \lfloor\frac{n-1}{2}\rfloor$.
    We let $\ell_1 = m$ and $\ell_2 = m + 1$; then, clearly $\ell_1 + \ell_2 \le n$.
    We first show that $\mathrm{lcm} \{ \ell_1, \ell_2 \} = m \cdot (m + 1)$,
    where the part $\mathrm{lcm} \{ \ell_1, \ell_2 \} \le m \cdot (m+1)$ is obvious.
    For the converse, let $a = \mathrm{lcm} \{ \ell_1, \ell_2 \}$ and let $c$ be such that $a = m \cdot c$.
    From $(m+1) | (m+1) \cdot c$ and $(m+1) | a$ we obtain $(m+1) | ((m+1) \cdot c - a) = c$,
    hence $m \cdot (m+1) | a$ and thus $a \ge m \cdot (m+1)$.

    The next step is to show that $m \cdot (m + 1) \ge n$.
    If $n = 5$, then $2 \cdot 3 = 6 \ge 5$.
    If $n \ge 6$, we note that $m \ge \frac{n - 2}2$ and $m + 1 \ge \frac{n}2$,
    hence $m \cdot (m + 1) \ge \frac{n - 2}2 \cdot \frac {n}2 = \frac{n \cdot (n - 2)}4 \ge n$.
    
    Finally, if $g(n) = m \cdot (m + 1)$, the proof is concluded with $k \ge 2$ and $\ell_2 = m + 1 \ge 2$.
    Otherwise, there exist $k' \in \mathbb N_+$ and $\ell_1', \dots, \ell_{k'}' \in \mathbb N_+$ such that $\ell_1' + \dots + \ell_{k'}' \le n$ and $\mathrm{lcm} \{ \ell_1', \dots, \ell_{k'}' \} = g(n)$.
    But then $\mathrm{lcm} \{ \ell_1', \dots, \ell_{k'}' \} > n$ and hence $k' \ge 2$ and there must be $i \in [k']$ with $\ell_i' \ge 2$.
\end{proof}

A \emph{prime number} is a natural number $a \ge 2$ such that no $1 < b < a$ divides $a$ (we note that $1 | a$ and $a | a$ hold for every $a \ge 2$).
For each $i \in \mathbb N$, we let $p_i$ be the $i$-th prime number,
i.e., $p_1 = 2$ and, for each $i \in \mathbb N$, $p_{i+1}$ is the smallest prime number that is greater than $p_i$.
For each real number $x$, the number of primes up to $x$ is denoted by $\pi(x)$, i.e., $\pi(x) = |\{ i \mid p_i \le x\}|$.
Ramanujan~\cite{Ramanujan1919} proved that, for each $i \in \mathbb{N}_+$, there is a natural
number $N$ such that $\pi(x) - \pi(\frac{x}{2}) \ge i$ for every $x \ge N$; by
well-ordering, there is a smallest such number, denoted $R_i$ and called the $i$-th
\emph{Ramanujan prime}~\cite{Sondow2009}. The first five Ramanujan primes are
$R_1, \dots, R_5 = 2, 11, 17, 29, 41$~\cite{Ramanujan1919}.

\section{Tree Automata and Synchronization}
\label{sec:sync-fta}

Just as we view trees as generalisations of strings, we can understand tree automata as generalisations of string automata. Recall that a \emph{deterministic finite-state (string) automaton} (DFA) is a tuple $\mathcal{A}=(Q,\Sigma,q_0,(\delta_\sigma)_{\sigma\in\Sigma},F)$ where $Q$ is a finite set of states, $\Sigma$ an alphabet, $q_0\in Q$ and $F\subseteq Q$ (initial state and set of final states), and $\delta_\sigma\colon Q\to Q$ is a transition function for each symbol $\sigma$. While DFA can read a string $w=a_1\ldots a_n$ from left to right by assigning (starting from $q_0$) the appropriate state to each symbol and resulting in a \emph{run} $(q_0,a_1,q_1)\ldots(q_{n-1},a_n,q_n)$, tree automata need to process the branching structure of a tree. Starting from the leaf symbols, again each symbol gets assigned a state. However, as a node may have several children, we need to incorporate several child states in our transition function and, thus, allow for functions of higher arity (cf.\@ Figure~\ref{fig:strings-to-trees}).

Formally, a \emph{bottom-up deterministic tree automaton} (DTA) is a tuple $\mathcal{A}=(Q,\Sigma,(\delta_\sigma)_{\sigma\in\Sigma},F)$ where now $\Sigma$ is a ranked alphabet and $\delta_\sigma\colon Q^{\mathrm{rk}(\sigma)}\to Q$ is the transition function for each symbol $\sigma\in\Sigma$ with its arity determined by the rank of $\sigma$. Note that we do not need an initial state in this setting, as $\delta_\alpha$ is a constant for each (nullary) leaf symbol $\alpha$. 
Moreover, since the question of tree acceptance and the language of $\mathcal{A}$ is not relevant for synchronization, we do not need the set of final states in this paper and simply choose $F = Q$.

\begin{figure}[t]
    \centering
    \includegraphics[width=0.9\linewidth]{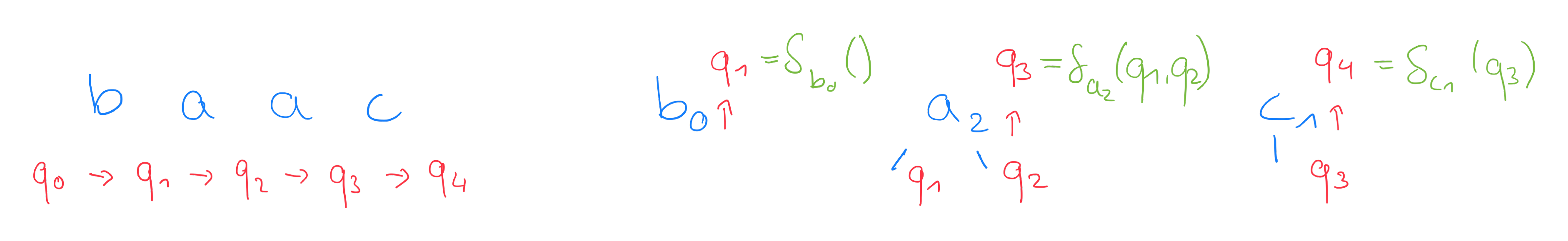}
    \caption{While in string automata transitions are unary functions allowing to read the string in a linear fashion from left to right (as in the left part), the transition functions of tree automata incorporate the arity of the related symbol (right part).}
    \label{fig:strings-to-trees}
\end{figure}

Investigating the synchronization of a DTA $\A$ benefits from additional tools to analyze the state behavior of $\A$ whilst reading a tree.
For this, we extend the transition functions to nonempty sets of input states.
Formally, let $\delta_\sigma(Q_1,\ldots,Q_{\mathrm{rk}(\sigma)})=\{\delta_\sigma(q_1,\ldots,q_{\mathrm{rk}(\sigma)})\mid q_1\in Q_1,\ldots,q_{\mathrm{rk}(\sigma)}\in Q_{\mathrm{rk}(\sigma)}\}$; that way we determine which set of states is reachable if we can input arbitrary states from $Q_1,\ldots,Q_{\mathrm{rk}(\sigma)}$, respectively, to $\delta_\sigma$. Furthermore, in order to determine the set of states $\mathcal{A}$ can reach after reading a tree $\xi\in\mathrm{T}_\Sigma(X)$, we define the function $\mathrm{Run_\A}\colon \mathrm{T}_\Sigma(X)\to 2^{Q}$ with 
\[\mathrm{Run}_\A(\xi)=
\begin{cases}
    Q & \text{if } \xi\in X\\
    \delta_\sigma(\mathrm{Run}_\A(\xi_1),\ldots,\mathrm{Run}_\A(\xi_n)) & \text{if } \xi=\sigma(\xi_1,\ldots,\xi_n)
\end{cases}\]
where we allow at each position in $\xi$, which carries a variable, the full set $Q$ of states as input. Note that, if $\xi$ does not contain a variable, then $\mathrm{Run}_\A(\xi)$ is a singleton as $\mathcal{A}$ is deterministic.

Before we introduce the notion of tree automata synchronization considered in this paper, let us briefly recap what synchronization means in the context of string automata. We say that a DFA $\mathcal{A}=(Q,\Sigma,q_0,(\delta_\sigma)_{\sigma\in\Sigma},F)$ is \emph{synchronizing} if there exists a string $w\in\Sigma^*$ which converts $\mathcal{A}$ into one fixed state $q$ no matter with which state $\mathcal{A}$ starts to process $w$. Clearly, not every string automaton is synchronizing but, given a DFA $\mathcal{A}$, the existence of such a synchronizing string for $\mathcal{A}$ can be tested in polynomial time \cite{Starke1966,Eppstein1990}. In the setting of tree automata, the definition of synchronizing trees is less clear than in the case of strings. As explained in detail in \cite{BlaJanPla26}, it makes a difference whether one considers as synchronizing tree (i) a tree without variables, (ii) a context with at least one variable, or (iii) a full context.
Since case (i) was shown to be trivial and case (ii) was reduced to the string case,
we will focus on case (iii) of full contexts, which were called \emph{strongly synchronizing} in~\cite{BlaJanPla26}.

\begin{definition}
    Let $\mathcal{A}=(Q,\Sigma,(\delta_\sigma)_{\sigma\in\Sigma},F)$ be a DTA and $\xi\in\mathrm{T}_\Sigma(X)$ a full context. We say that $\xi$ is a \emph{synchronizing term (for $\mathcal{A}$)} if $|\mathrm{Run}_\A(\xi)|=1$.
\end{definition}

A tree automaton $\mathcal{A}$ is called \emph{synchronizing} if it has a synchronizing term.
A synchronizing term of smallest height is called the \emph{smallest synchronizing term} of $\A$ and its height is called the \emph{reset threshold} of $\A$, denoted by $\mathrm{rt}(\A)$.
In this paper, we are interested in the question of how large the reset threshold of a tree automaton can be depending on its number of states. Whilst in \cite{BlaJanPla26} a super-polynomial lower bound could only be shown in the case of an alphabet growing at least linearly, the construction we present in the following sections requires only two alphabet symbols: the binary symbol $f$ and the unary symbol $w$. For this, we will construct in Section \ref{sec:construction} two transition functions $\delta_f$ and $\delta_w$ that force the associated tree automaton $\A$ (defined in Section \ref{sec:bounds}) to read very high trees before it synchronizes.

\begin{remark}
    The way in which we have defined tree automata (the transitions are given by a family of operations on the set of states of varying arity, i.e., functions of the form $Q^{k_i}\to Q$) allows for an alternative view of tree automata in the context of synchronization: Instead of asking how large a synchronizing term must be, we can also ask how deeply we need to compose functions from $(\delta_\sigma)_{\sigma\in\Sigma}$ such that the resulting expression is a constant function which maps all inputs to a singleton. Therefore, in Section \ref{sec:construction} we will abstract from the tree automaton perspective and only investigate two particular functions $\delta_w\colon [n] \to [n]$ and $\delta_f\colon [n]^2 \to [n]$. In Section \ref{sec:bounds}, we will then apply the insights gained to the synchronization problem of tree automata.
\end{remark}

\section{Construction of $\delta_f$ and $\delta_w$} \label{sec:construction}

In this section we design two functions $\delta_f\colon [n]^2 \to [n]$ and $\delta_w\colon [n] \to [n]$ with the goal that the smallest composite expression over $\{\delta_f, \delta_w\}$ that maps $[n]$ to $\{1\}$ is as deeply nested as possible.
We will approach this by parameterizing $\delta_w$ and $\delta_f$ by positive numbers $\ell_1, \dots, \ell_k$ with $\ell_1 + \cdots + \ell_k \le n$, which will be fixed in Section \ref{sec:bounds}. Then we will show that, starting from $[n]$, the earliest singleton is reachable after exactly $\mathrm{lcm} \{\ell_1, \dots, \ell_k\} + 1$ nested applications of $\delta_w$ and $\delta_f$. The idea behind the $\ell_i$'s is that we want to partition $[n]$ into $k$ intervals, within which the application of $\delta_w$ results in a cyclic shift. This way, we can apply $\delta_w$ up to $\mathrm{lcm}\{\ell_1, \dots, \ell_k\}-1$ times consecutively to some initial set $A_0$ without running into a cycle. The job of mapping the full set $[n]$ to the appropriate $A_0$ will be done by the function $\delta_f$. Finally, $\delta_f$ will project $\delta_w^{\mathrm{lcm}\{\ell_1, \dots, \ell_k\}-1}(A_0)=F$ to $\{1\}$. An illustration of the resulting expression can be found in Figure \ref{fig:expression}.

\begin{figure}[t]
    \centering
    \includegraphics[width=0.35\linewidth]{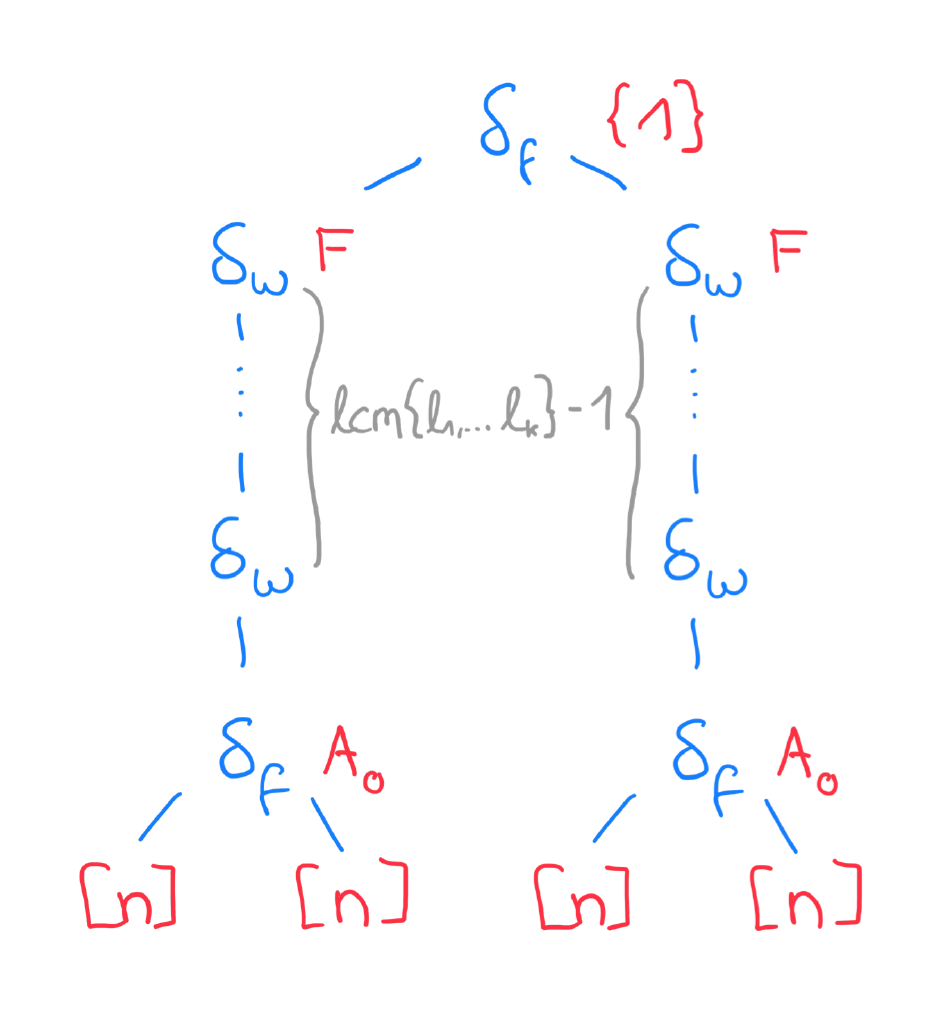}
    \caption{The structure of the smallest composite expression over $\{\delta_f,\delta_w\}$ that maps $[n]$ to $\{1\}$ together with the sets it reaches after individual application steps.}
    \label{fig:expression}
\end{figure}

Formally, let $n, k, \ell_1, \dots, \ell_k \in \mathbb N_+$ such that $\ell_1 + \cdots + \ell_k \le n$. We let $c_1 = 1$ and, for every $i \in [k]$, we let $c_{i+1} = c_i + \ell_i$. Moreover, for each $1 \leq i \leq k$ we let $L_i=\{c_i,\ldots,c_{i+1}-1\}$ be the $i$-th block (of length $\ell_i$) and we let $L_{\mathrm{const}}=\{c_{k+1},\ldots,n\}$ be the residual interval (in case $\ell_1 + \cdots + \ell_k < n$).

Moreover, we let
\begin{align*}
    A_0 &= \{c_1, c_2, \dots, c_{k}\} \cup \{ c_{k+1}, c_{k+1} + 1, \dots, n \} \\
    F &= \{c_2 - 1, c_3 - 1, \dots, c_{k+1} - 1\} \cup \{ c_{k+1}, {c_{k+1} + 1}, \dots, n \}
\end{align*}
and we define the mappings $\delta_w\colon [n] \to [n]$ and $\delta_f\colon [n]^2 \to [n]$ as follows:
\begin{align*}
\delta_w(q) &=
\begin{cases}
    c_{i} & \text{if } q=c_{i+1}-1 \text{ for some } i\in[k] \text{ ($q$ is the last element of a block)}\\
    q+1 & \text{if } q\in L_i\setminus\{c_{i+1}-1\} \text{ for some } i\in[k]\\
    q & \text{if } q\in L_{\mathrm{const}}\\
\end{cases}\\
\intertext{for each $q\in[n]$, and}
\delta_f(q_1,q_2)&=
\begin{cases}
    1 & \text{if } q_1,q_2 \in F\\
    c_i & \text{if } q_1\in L_i \text{ for some } i\in[k], (q_1,q_2)\notin F^2\\
    q_1 & \text{else}
\end{cases}
\end{align*}
for each $q_1,q_2\in[n]$. An example for $\delta_w$ and $\delta_f$ as well as their functionality is given in Figure \ref{fig:functions}.

\begin{figure}[t]
    \centering
    \includegraphics[width=0.5\linewidth]{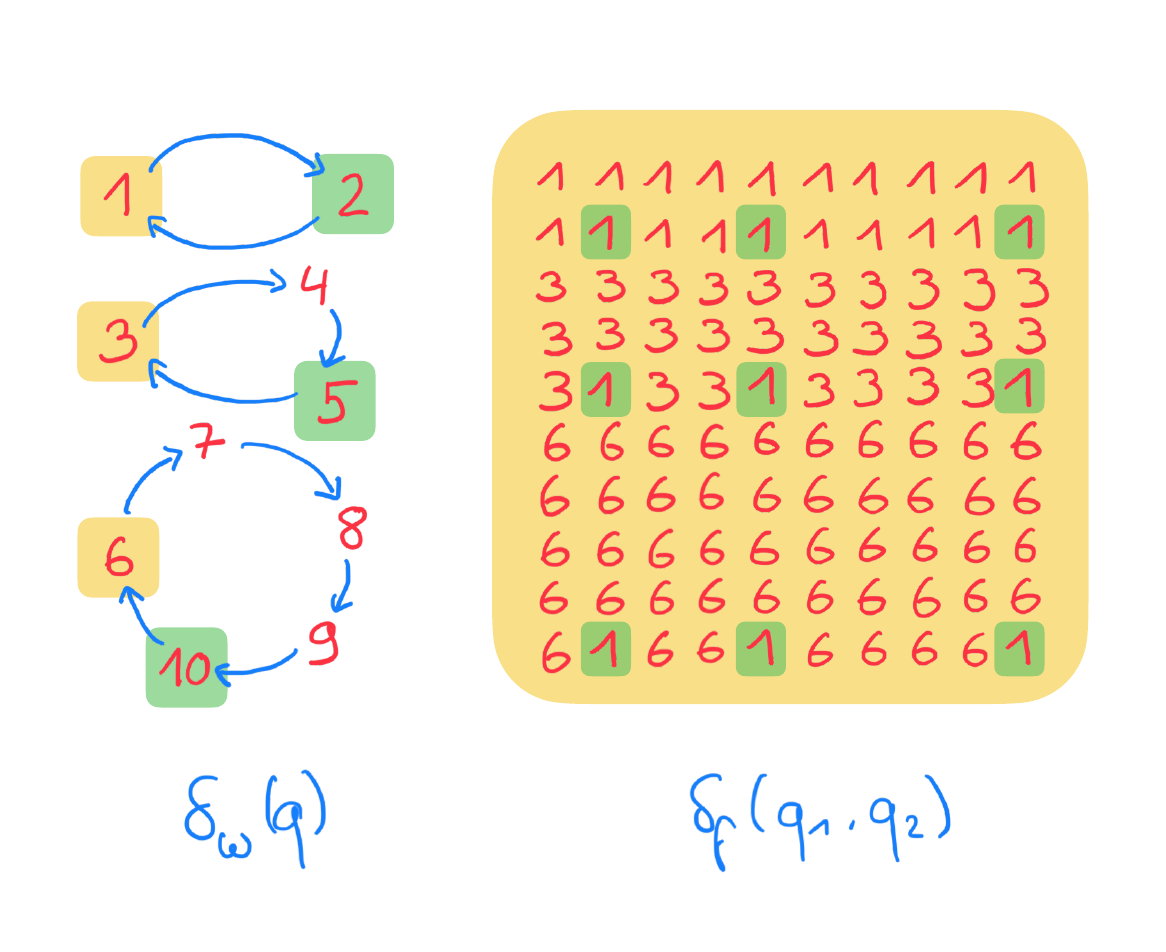}
    \caption{The shape of $\delta_w$ and $\delta_f$ for $n=10$ and $\ell_1=2$, $\ell_2=3$, $\ell_3=5$. Starting from $\delta_f([n],[n])=\{1,3,6\}$ (highlighted in yellow), it takes $\mathrm{lcm}\{\ell_1, \dots, \ell_3\}-1=2\cdot 3\cdot 5-1$ steps for $\delta_w$ to reach $\{2,5,10\}$, the last set before running into a cycle (highlighted in green). This very set is mapped by $\delta_f$ to $\{1\}$.}
    \label{fig:functions}
\end{figure}

In the following, we will show that our intuition regarding the functionality of $\delta_w$ and $\delta_f$ is correct, i.e., the earliest singleton is reachable after exactly $\mathrm{lcm} \{\ell_1, \dots, \ell_k\} + 1$ nested applications (cf.\@ Lemma~\ref{lem:smallest-h}).
As a preparation for this, we will investigate how these functions behave when we apply them repeatedly, starting from the set $[n]$.

From now on, we let $\mathsf{LCM}$ denote the least common multiple of $\{\ell_1, \dots, \ell_k\}$.
Moreover, for the correctness of the statements in this section, we prohibit the corner cases that $\ell_1 = \dots = \ell_k = 1$ (which would entail $A_0 = F = [n]$, breaking Lemma~\ref{lem:reach-induction}) and $|A_0| = 1$ (which would break Lemma~\ref{lem:smallest-h}; note that this happens if and only if both $k = 1$ and $L_{\mathrm{const}} = \emptyset$).

\subsubsection*{Properties of $\delta_w$}

First let us note that $\delta_w$ applied to the entire set $[n]$ is simply the identity, which follows directly from its definition.

\begin{observation}\label{obs:image-delta-w}
    $\delta_w([n])=[n]$.
\end{observation}

Since every block length $\ell_i$ divides $\mathsf{LCM}$, we obtain that applying $\delta_w$ $\mathsf{LCM}$-times subsequently to $A_0$ results in a cycle.
Then it is easy to see that applying it one time less yields $F$.

\begin{observation}\label{obs:delta-w-F}
    $\delta_w^{\mathsf{LCM}}(A_0)=A_0$ and $\delta_w^{\mathsf{LCM}-1}(A_0) = F$.
\end{observation}

Moreover, as the application of $\delta_w$ just shifts the elements within the $k$ intervals we consider cyclically, each $\delta_w^i(A_0)$ contains exactly one element from every interval.

\begin{lemma}\label{lem:delta-w-chain}
    For every $i \in \mathbb N$ and $j \in [k]$ it holds that $|\delta_w^i(A_0) \cap L_j| = 1$.
    Moreover, $L_{\mathrm{const}} \subseteq \delta_w^i(A_0)$.
\end{lemma}

\begin{proof}
    The proof is done by induction on $i$ where the induction base ($i=0$) follows directly from the definition of $A_0$.
    For the induction step, let $j \in [k]$.
    By the induction hypothesis, there is exactly one $a \in L_j$ such that $a \in \delta_w^i(A_0)$.
    Then, by definition, we have $\delta_w(a) \in L_j$.
    Moreover, by induction hypothesis, for every $a \in L_{\mathrm{const}}$ we have $a \in \delta_w^i(A_0)$.
    Then, by definition, $\delta_w(a) = a$.
\end{proof}

Finally we can show that, with every consecutive application of $\delta_w$ below $\mathsf{LCM}$, we obtain a previously unseen set.

\begin{lemma}\label{lem:delta-i-neq-delta-j}
    For every $i, j \in \{0, \dots, \mathsf{LCM}-1\}$ it holds that $i \not= j$ implies $\delta_w^i(A_0) \not= \delta_w^j(A_0)$.
\end{lemma}

\begin{proof}
    Let $i, j \in \{0, \dots, \mathsf{LCM}-1\}$ such that $i > j$ (w.l.o.g.) and assume that $\delta_w^i(c_h) = \delta_w^j(c_h)$ for every $h \in [k]$.
    By the definition of $\delta_w$, this implies $\ell_h | (i - j)$ for every $h \in [k]$, and as a consequence $\mathsf{LCM} | (i - j)$.
    But since $i, j \in \{0, \dots, \mathsf{LCM}-1\}$, this implies $i = j$.
    Hence we can assume that there exists $h \in [k]$ such that $\delta_w^i(c_h) \not= \delta_w^j(c_h)$.
    By Lemma~\ref{lem:delta-w-chain}, $\delta_w^i(c_h)$ and $\delta_w^j(c_h)$ are the unique elements of $L_h$ in $\delta_w^i(A_0)$ and $\delta_w^j(A_0)$, respectively, so $\delta_w^i(A_0) \not= \delta_w^j(A_0)$.
\end{proof}

\subsubsection*{Properties of $\delta_f$}

Turning to the definition of $\delta_f$, it is obvious that applying it to $F$ yields the singleton $\{1\}$.

\begin{observation}\label{obs:delta-f-image}
    $\delta_f(F, F) = \{1\}$.
\end{observation}

Moreover, all other sets reachable from $A_0$ via $\delta_w$ are mapped by $\delta_f$ back to $A_0$. It follows that $\delta_f$ can not be used as an abbreviation: only after applying $\delta_w$ $\mathsf{LCM}-1$ times we can use $\delta_f$ in order to reach $\{1\}$.

\begin{lemma}\label{lem:image-delta-f}
    For every $A, B \in \{[n]\} \cup \delta_w^*(A_0)$ the following holds: if $B \not\subseteq F$, then $\delta_f(A, B) = A_0$.
\end{lemma}

\begin{proof}
    Since $\delta_f(A, B) = \bigcup_{x \in B} \delta_f(A, \{x\})$, we let $x \in B$ and distinguish two cases.
    \begin{description}
    \item[$x \in B \setminus F$:]
        If $A \not= [n]$, then, by Lemma~\ref{lem:delta-w-chain}, for every $j \in [k]$ there is exactly one $h \in A$ with $h \in L_j$.
        By definition, $\delta_f(h, x) = c_j$.
        Also, by Lemma~\ref{lem:delta-w-chain}, it holds that $L_{\mathrm{const}} \subseteq A$, and by definition, $\delta_f(L_{\mathrm{const}}, \{x\}) = L_{\mathrm{const}}$.
        Since $A \subseteq L_1 \cup \cdots \cup L_k \cup L_{\mathrm{const}}$, we obtain $\delta_f(A, \{x\}) = A_0$.
        If $A = [n]$, then $\delta_f([n], \{x\}) = A_0$ is easy to see.
    \item[$x \in B \cap F$:]
        By definition, for every $q \in [n]$, we have $\delta_f(q, x) = 1$ if $q \in F$ and $\delta_f(q, x) = c_i$ if $q \in L_i \setminus F$ for some $i \in [k]$.
        Since $c_1 = 1$, we obtain $\delta_f(A, \{x\}) \subseteq A_0$.
    \end{description}
    As $B \not\subseteq F$, there exists at least one $x \in B \setminus F$, and thus $\delta_f(A, B) = A_0$.
\end{proof}

\subsubsection*{Reachability of singletons}

Next, we define the sets $\mathrm{Reach}_i$ of subsets of $[n]$ that are reachable within $i$ applications of $\delta_w$ and $\delta_f$ starting from $[n]$. We let $\mathrm{Reach}_0=\{[n]\}$ and 
\[\mathrm{Reach}_i=\mathrm{Reach}_{i-1}\cup\{\delta_f(A,B),\delta_w(A)\mid A,B\in\mathrm{Reach}_{i-1}\}.\]
Clearly, there is a composite expression over $\delta_f$ and $\delta_w$ of height $i$ that is constant if and only if $\mathrm{Reach}_i$ contains a singleton.

\begin{lemma}\label{lem:reach-induction}
    For each $i\in \{1,\ldots,\mathsf{LCM}\}$ we have $\mathrm{Reach}_i=\{[n],A_0\}\cup\delta^{\leq i-1}_w(A_0)$.
\end{lemma}

\begin{proof}
The proof is by induction on $i$. For the induction base, let $i=1$. Then we obtain
\begin{align*}
    \mathrm{Reach}_1=&\mathrm{Reach}_{0}\cup\{\delta_f(A,B),\delta_w(A)\mid A,B\in\mathrm{Reach}_{0}\}\\
    =& \{[n]\}\cup\{\delta_f([n],[n]),\delta_w([n])\}\\
    =& \{[n],A_0\} \tag{Observation \ref{obs:image-delta-w} and Lemma \ref{lem:image-delta-f}}\\
    =& \{[n], A_0\} \cup \delta_w^{\le0}(A_0)
\end{align*}
where Lemma~\ref{lem:image-delta-f} yields $\delta_f([n], [n]) = A_0$ since $[n] \supset F$ (due to our assumption that there exists $j \in [k]$ with $\ell_j > 1$).
For the induction step, let $i\in \{1,\ldots,\mathsf{LCM}-1\}$. Then
\begin{align*}
    \mathrm{Reach}_{i+1}=&\mathrm{Reach}_{i}\cup\{\delta_f(A,B),\delta_w(A)\mid A,B\in\mathrm{Reach}_{i}\}\\
    =& \{[n],A_0\}\cup\delta^{\leq i-1}_w(A_0)\cup\{\delta_f(A,B),\delta_w(A)\mid A,B\in\{[n],A_0\}\cup\delta^{\leq i-1}_w(A_0)\}\tag{IH}\\
    =& \{[n],A_0\}\cup\delta^{\leq i-1}_w(A_0)\cup\{\delta_f(A,B),\delta_w(A)\mid A,B\in\{[n],A_0\}\cup\delta^{\leq i-1}_w(A_0) \setminus \{F\} \}\tag{*}\\
    =& \{[n],A_0\}\cup\delta^{\leq i-1}_w(A_0)\cup\{\delta_w(A)\mid A\in\{[n],A_0\}\cup\delta^{\leq i-1}_w(A_0)\}\tag{**}\\
    =& \{[n],A_0\}\cup\delta^{\leq i}_w(A_0)\tag{Observation \ref{obs:image-delta-w}}
\end{align*}
where $(*)$ holds by Observation \ref{obs:delta-w-F} and Lemma \ref{lem:delta-i-neq-delta-j} and $(**)$ holds by Lemma \ref{lem:image-delta-f}; we note that by Lemma~\ref{lem:delta-w-chain}, no proper subset of $F$ can be an element of $\delta^{\leq i-1}_w(A_0)$.
\end{proof}

\begin{lemma}\label{lem:smallest-h}
    The smallest $h$ such that $\mathrm{Reach}_h$ contains a singleton is $h=\mathsf{LCM}+1$.
\end{lemma}

\begin{proof}
    By Lemma~\ref{lem:reach-induction}, we have $\mathrm{Reach}_i=\{[n],A_0\}\cup\delta^{\leq i-1}_w(A_0)$ for every $1 \le i \le \mathsf{LCM}$ (we note that Lemma~\ref{lem:reach-induction} can be applied by our assumption that there exists $j \in [k]$ with $\ell_j > 1$).
    Moreover, Lemma~\ref{lem:delta-w-chain} together with the assumption that either $k \ge 2$ or $L_{\mathrm{const}} \not= \emptyset$ implies that no set in $\mathrm{Reach}_i$ is a singleton.
    Also, since $\delta_w^{\mathsf{LCM}-1}(A_0) = F$ by Observation \ref{obs:delta-w-F}, we have $F \in \mathrm{Reach}_{\mathsf{LCM}}$.
    Then, by Observation~\ref{obs:delta-f-image}, we have $\{1\} \in \mathrm{Reach}_{\mathsf{LCM+1}}$.
\end{proof}

\begin{remark}\label{rem:high-rank-gain}
    We note that the role of $\delta_f$ cannot be fulfilled by a unary function.
    This is because, in essence, when considered as a function on sets, $\delta_f$ combines two unary functions and switches between them based on its first argument.
    These functions are $a\colon [n] \to [n]$ and $b \colon [n] \to [n]$ where
    \[
        a(q) = \begin{cases}
            c_i & \text{if } q\in L_i \text{ for some } i\in[k]\\
            q & \text{else}
        \end{cases}
        \qquad\text{and}\qquad
        b(q) = \begin{cases}
            1 & \text{if } q \in F\\
            c_i & \text{if } q\in L_i\setminus F \text{ for some } i\in[k]\\
            q & \text{else.}
        \end{cases}
    \]
    Then, for every $A \subseteq [n]$ and $B \subseteq [n] \setminus F$, we have $\delta_f(A, B) = a(A)$ and $\delta_f(A, F) = b(A)$.
    This technique allows us to defer the application of the function $b$ until the set $F$ is reached (which takes $\mathsf{LCM}$ steps in our case), using the non-productive function $a$ in the meantime.

    Since a set of unary functions alone cannot behave in such a way, the technique of this section cannot be applied to DFA.
    Indeed, whenever two states of a DFA can be synchronized, this can be done in at most $n^2$ applications of any transition function (cf.\@ e.g.~\cite{Eppstein1990}); our $\mathsf{LCM}$ may exceed $n^2$ by far.
\end{remark}

\section{Lower bounds on smallest synchronizing terms}\label{sec:bounds}

In this section we show that, given $\ell_1, \dots, \ell_k$, we can use the functions $\delta_w$ and $\delta_f$ from the previous section to construct a DTA whose smallest synchronizing term has height $\mathrm{lcm} \{\ell_1, \dots, \ell_k\} + 1$ (Lemma~\ref{lem:lcm-automaton}).
We apply this general result to two particular strategies for choosing $\ell_1, \dots, \ell_k$ in order to obtain improved lower bounds for the reset thresholds of DTA over two-symbol alphabets: one polynomial with a freely choosable degree (Theorem~\ref{thm:polynomial}) and one subexponential (Theorem~\ref{thm:landau}).

Our DTA construction uses $[n]$ as the set of states and $\{ w, f \}$ as the alphabet.
This allows us to define its transition function in terms of $\delta_w$ and $\delta_f$ from the previous section.
Thus, there is a synchronizing term of height $h$ over $\{ w, f \}$ if and only if $\mathrm{Reach}_h$ contains a singleton.

\begin{lemma}\label{lem:lcm-automaton}
    Let $n, k, \ell_1, \dots, \ell_k \in \mathbb N_+$ such that $\ell_1 + \cdots + \ell_k \le n$, there exists $i \in [k]$ with $\ell_i \ge 2$, and $k \ge 2$.%
    \footnote{We note that this lemma would still hold if $k = 1$ but $\ell_1 < n$. However, since this case is not relevant for our main theorems, we have neglected it in order to make the statement of this lemma more concise.}
    Then there exists a DTA with $n$ states whose smallest synchronizing term has height $\mathrm{lcm} \{\ell_1, \dots, \ell_k\} + 1$.
\end{lemma}

\begin{proof}
    We let $\A = ([n], \{w^{(1)}, f^{(2)}\}, (\delta_\sigma)_{\sigma\in\{w,f\}}, [n])$ with $\delta_w$ and $\delta_f$ defined as in the previous section.
    We have 
    \begin{align*}
        \mathrm{rt}(\A) &= \min \{ \height(t) \mid t \in \mathrm{T}_{\Sigma}(X) \text{ is a full context}, |\mathrm{Run}_{\A}(t)| = 1 \} \\
        &= \min \{ h \in \mathbb N \mid \text{$\mathrm{Reach}_h$ contains a singleton}\} \tag{*} \\
        &= \mathrm{lcm} \{\ell_1, \dots, \ell_k\} + 1 \tag{Lemma~\ref{lem:smallest-h}}
    \end{align*}
    where ($*$) can be easily shown by structural induction.
\end{proof}

\begin{remark}
    We note that the tree automaton constructed above does not incorporate nullary symbols and, thus, the tree language it accepts is empty.
    While it is not relevant to synchronization, our construction can easily be transferred to automata that read leaf symbols.
    This will not affect our results on reset thresholds since synchronizing terms do not contain nullary symbols by definition.
\end{remark}

Our strategy for constructing a family of DTA with a polynomial reset threshold is as follows.
Fix $d \ge 2$.
For any sufficiently large $n$, pick the $d$ largest prime numbers up to $\lfloor \frac{n}{d} \rfloor$.
As their least common multiple is in $\Omega\big(\left(\frac{n}{d}\right)^d\big)$, Lemma~\ref{lem:lcm-automaton} yields a DTA with a polynomial reset threshold of degree~$d$.

\begin{theorem}\label{thm:polynomial}
    Let $d \ge 2$.
    For every $n \ge d \cdot R_d$ there exists a DTA with $n$ states whose smallest synchronizing term has height at least $(\frac{n}{4d})^d$.
\end{theorem}

\begin{proof}
    For every $n \ge d \cdot R_d$, we let $m = \lfloor n/d \rfloor$ and
    \[
        (\ell_1, \dots, \ell_d) = (p_{\pi(m)-d+1}, p_{\pi(m)-d+2}, \dots, p_{\pi(m)})\text.
    \]
    We note that, since $n \ge d \cdot R_d$, we have $m \ge R_d$ and thus $\pi(m) - \pi(\tfrac m2) \ge d$,
    i.e., there are at least $d$ prime numbers greater than $\tfrac m2$ and at most $m$.
    Hence $(\ell_1, \dots, \ell_d)$ exists and $\ell_i > \tfrac m2$ for every $i \in [d]$.
    Moreover, $\sum_{i=1}^d \ell_i \le d \cdot m = d \lfloor n/d \rfloor \le n$.

    By Lemma~\ref{lem:lcm-automaton} (which can be applied with $k = d \ge 2$ and since each $\ell_i$ is prime and hence $\ge 2$), there exists a DTA whose smallest synchronizing term has height $\mathrm{lcm} \{\ell_1, \dots, \ell_d\} + 1$.
    Now, since $(\ell_1, \dots, \ell_d)$ are pairwise distinct prime numbers, we have $\mathrm{lcm} \{\ell_1, \dots, \ell_d\} = \prod_{i=1}^d \ell_i$.
    Moreover, $n/d \ge R_d \ge 1$ and hence $m = \lfloor n/d \rfloor \ge \frac{n}{2d}$, since $\lfloor x \rfloor \ge x/2$ for every real $x \ge 1$.
    This implies
    \[
        \ell_i > \frac{m}2 \ge \frac{n}{4d}.
    \]
    As a consequence, $\prod_{i=1}^d \ell_i \ge \left( \frac{n}{4d} \right)^d$.

    We obtain that the height of the smallest synchronizing term of the DTA is
    \[
        \mathrm{lcm} \{\ell_1, \dots, \ell_d\} + 1 = \prod_{i=1}^d \ell_i + 1 > \left( \frac{n}{4d} \right)^d.
        \qedhere
    \]
\end{proof}

\begin{remark}
    We note that our lower bound substantially underestimates the height of the smallest synchronizing term.
    Since $R_d < p_{3d}$~\cite{Laishram2010} and $p_{3d} < 3d(\log 3d + \log\log 3d)$~\cite{RosserSchoenfeld1962},
    it covers all $n \in \Omega(d^2 \log d)$.
    While a more elaborate selection of $d$ pairwise coprime numbers (that are not necessarily primes) could lead to a better bound that potentially covers more initial $n$,
    we stick to our strategy since it is conceptually simple and fulfills the goal of showing polynomial lower bounds of arbitrary degree.
    Our best bound (which we show next) will be superpolynomial and thus exceed any polynomial construction anyway.
\end{remark}

Our strategy for constructing a family of DTA with a subexponential reset threshold is as follows.
For any $n$, pick $k$ and $\ell_1, \dots, \ell_k$ such that $\mathrm{lcm} \{\ell_1, \dots, \ell_k\}$ is maximized and $\ell_1 + \cdots + \ell_k \le n$.
As their least common multiple is $e^{(1+o(1)) \cdot \sqrt{n \cdot \ln n}}$, Lemma~\ref{lem:lcm-automaton} yields a DTA with a subexponential reset threshold.

\begin{theorem}\label{thm:landau}
    For every $n \ge 5$ there exists a DTA $\mathcal L_n$ which has $n$ states and a smallest synchronizing term of height $g(n) + 1$, where $g\colon \mathbb N \to \mathbb N$ is Landau's function.
    Thus, the height of the smallest synchronizing term of the family $(\mathcal L_n \mid n \ge 5)$ has growth $e^{(1+o(1)) \cdot \sqrt{n \cdot \ln n}}$.
\end{theorem}

\begin{proof}
    Let $n \ge 5$.
    By Lemma~\ref{lem:lower-bound-landau}, there exist $k \ge 2$ and $\ell_1, \dots, \ell_k \in \mathbb N_+$ such that $\ell_1 + \dots + \ell_k \le n$, $\mathrm{lcm} \{ \ell_1, \dots, \ell_k \} = g(n)$, and there is some $i \in [k]$ with $\ell_i \ge 2$.
    By Lemma~\ref{lem:lcm-automaton} there exists a DTA $\mathcal L_n$ with $n$ states whose reset threshold is $\mathrm{lcm} \{\ell_1, \dots, \ell_k\} + 1 = g(n) + 1$.
    By~\cite{Lan1903}, $g(n) = e^{(1+o(1)) \cdot \sqrt{n \cdot \ln n}}$, which yields the claimed growth.
\end{proof}

\begin{remark}
    We note that the search space for optimal $\ell_1, \dots, \ell_k$ is restricted to powers of pairwise distinct primes.
    Then, the goal becomes partitioning $n$ into prime powers (and potentially some slack) with a maximal product.
    This can be solved by Dynamic Programming (e.g., a modification of the $O\big(\frac{n^{3/2}}{\sqrt{\log n}}\big)$ algorithm of~\cite{Nicolas1969} with backtracking information that allows retrieval of the prime powers).
\end{remark}

\section{Outlook}

We have improved the lower bound on the reset threshold of DTA over constant-sized alphabets from quadratic to subexponential.
This significantly narrows the gap towards the upper bound, which still sits at $2^n - n - 1$.
Our technique was flexible enough to also allow us to construct families of DTA that have (weaker) polynomial thresholds, one family for each degree greater than $1$.

However, the construction of this paper cannot be used to improve this bound further:
in essence, it iterates a unary function as long as possible (the binary function serves only as a guard), and by the very definition of Landau's function, the length of the iteration is already maximized.
We believe that the value of our construction lies in its minimality:
by using only two symbols and only one property exclusive to higher-arity functions (cf.\@ Remark~\ref{rem:high-rank-gain}), we are already able to surpass the bounds of DFA by far.
It is ongoing work to explore the possibilities enabled by binary functions further in order to reach an exponential lower bound (as it has already been done for alphabets that grow with the number of states).

\bibliographystyle{plain}
\bibliography{lit}

\end{document}